\documentclass[submission,copyright,creativecommons]{eptcs}
\providecommand{\event}{LSFA 2011}

\usepackage{bnf}
\usepackage{mathpartir}
\usepackage{amsmath}
\usepackage{amssymb}
\usepackage{stmaryrd}
\usepackage{float}
\usepackage{prooftree}
\usepackage[all]{xy}

\usepackage{color}
\usepackage{soul}
\usepackage{hyperref}
\hypersetup{
	colorlinks=true,
	citecolor=blue,
	linkcolor=black,
    urlcolor=black,
	pdftitle={Confluence via strong normalisation in an algebraic lambda-calculus with rewriting},
    pdfauthor={P. Buiras, A. D\'iaz-Caro, M. Jaskelioff},
	hypertexnames=false		
}
\usepackage{breakurl}
\floatstyle{boxed}
\restylefloat{figure}
\restylefloat{table}

\newtheorem{theorem}{Theorem}[section]
\newtheorem{lemma}{Lemma}[section]
\newtheorem{corollary}[theorem]{Corollary}

\newenvironment{proof}[1][Proof]{\begin{trivlist}
\item[\hskip \labelsep {\bfseries #1}]}{\end{trivlist}}
\newenvironment{definition}[1][Definition]{\begin{trivlist}
\item[\hskip \labelsep {\bfseries #1}]}{\end{trivlist}}
\newenvironment{example}[1][Example]{\begin{trivlist}
\item[\hskip \labelsep {\bfseries #1}]}{\end{trivlist}}
\newenvironment{remark}[1][Remark]{\begin{trivlist}
\item[\hskip \labelsep {\bfseries #1}]}{\end{trivlist}}

\newcommand{\qed}{\hfill \mbox{\raggedright \rule{.07in}{.1in}}} 

\newcommand{\llin}{\ensuremath{\lambda_{{\it lin}}}}
\newcommand{\lalg}{\ensuremath{\lambda_{{\it alg}}}}
\newcommand{\additive}{\ensuremath{{\it Additive}}}
\newcommand{\CA}{\ensuremath{\lambda_{\it CA}}}
\newcommand{\ladd}{\ensuremath{\lambda^{\!\textrm{add}}}}
\newcommand{\cf}{\emph{cf.}~}
\newcommand{\ie}{\emph{i.e.}~}
\newcommand{\eg}{\emph{e.g.}~}
\newcommand{\ve}[1]{\ensuremath{\mathbf{#1}}}
\newcommand{\type}{\colon\!}
\newcommand{\pair}[2]{\ensuremath{\langle #1,#2\rangle}}
\newcommand{\thesi}{\vdash_{\scriptstyle F}}
\newcommand{\uno}{\ensuremath{\mathbf{1}}}
\newcommand{\pce}{\ensuremath{\preccurlyeq}}
\newcommand{\sce}{\ensuremath{\succcurlyeq}}

\newcommand{\vdashadd}{\ensuremath{\vdash_{\!\!A}}}
\newcommand{\toA}{\ensuremath{\rightarrow_{{}_{\!A}}}}
\newcommand{\toF}{\ensuremath{\rightarrow_{{}_{\!F}}}}
\newcommand{\toAnormal}{\ensuremath{\fnormalA{}}}
\newcommand{\tonormal}{\ensuremath{\fnormal{}}}
\newcommand{\sqsAstrict}{\ensuremath{\sqsubseteq}}
\newcommand{\sqsA}{\ensuremath{\lesssim}}
\newcommand{\sqsFp}{\ensuremath{\lessapprox}}
\newcommand{\sqsFpstrict}{\ensuremath{\sqsubseteq}_F}
\newcommand{\fnormal}[1]{\ensuremath{{#1}\!\!\downarrow}}
\newcommand{\fnormalA}[1]{\ensuremath{{#1}\!\!\downarrow_{{}_{\!A}}}}
\newcommand{\fnormalFp}[1]{\ensuremath{{#1}\!\!\downarrow_{{}_{\!F}}}}
\newcommand{\Fp}{\ensuremath{F_p}}

\newcommand{\toFpnormal}{\ensuremath{\downarrow_{{}_{\!F}}}}
\newcommand{\AddtoFp}[2]{\ensuremath{[#1]_{\mathsf{#2}}}}
\newcommand{\AddtoFpD}[1]{\ensuremath{\AddtoFp{#1}{D}}}
\newcommand{\TFA}{\ensuremath{T_{\CA}}}
\newcommand{\Tadd}{\ensuremath{T_{\ladd}}}
\newcommand{\TFp}{\ensuremath{T_{\Fp}}}

\title{\texorpdfstring{Confluence via strong normalisation in\\ an algebraic $\lambda$-calculus with rewriting}{Confluence via strong normalisation in an algebraic lambda-calculus with rewriting}}

\author{Pablo Buiras
\institute{Universidad Nacional de Rosario, FCEIA\\
Pellegrini 250\\ S2000BTP Rosario, SF, Argentina}
\email{pablo.buiras@gmail.com}
\and Alejandro D\'iaz-Caro
\institute{Universit\'e de Grenoble, LIG,\\
220 rue de la Chimie\\ 38400 Saint Martin d'H\`eres, France}
\institute{LIPN -- UMR CNRS 7030\\
Institut Galil\'ee - Universit\'e Paris-Nord\\
99, avenue Jean-Baptiste Cl\'ement\\
93430 Villetaneuse, France}
\email{alejandro@diaz-caro.info}
\and Mauro Jaskelioff
\institute{Universidad Nacional de Rosario, FCEIA\\
Pellegrini 250\\ S2000BTP Rosario, SF, Argentina}
\institute{CIFASIS\\
27 de Febrero 210 bis\\ S2000EZP Rosario, SF, Argentina}
\email{mauro@fceia.unr.edu.ar}
}

\def\titlerunning{Confluence via strong normalisation in an algebraic $\lambda$-calculus with rewriting}
\def\authorrunning{P. Buiras, A. D\'iaz-Caro and M. Jaskelioff}

\begin{document}
\maketitle
\begin{abstract}
  The linear-algebraic $\lambda$-calculus and the algebraic
  $\lambda$-calculus are untyped $\lambda$-calculi extended with
  arbitrary linear combinations of terms. The former presents the
  axioms of linear algebra in the form of a rewrite system, while the
  latter uses equalities. When given by rewrites, algebraic
  $\lambda$-calculi are not confluent unless further restrictions are added. 
  We provide a type system for the linear-algebraic
  $\lambda$-calculus enforcing strong normalisation, which gives back
  confluence. 
  The type system allows an abstract interpretation in System F.
\end{abstract}

\section{Introduction}\label{sec:intro}
Two algebraic versions of $\lambda$-calculus arose independently in
different contexts: the linear-algebraic $\lambda$-calculus
(\llin)~\cite{ArrighiDowekRTA08} and the algebraic $\lambda$-calculus
(\lalg)~\cite{VauxMSCS09}. The former was first introduced as a
candidate $\lambda$-calculus for quantum computation; a linear
combination of terms reflects the phenomenon of superposition, \ie
the capacity for a quantum system to be in two or more states at the
same time. The latter was introduced in the context of linear logic,
as a fragment of the differential
$\lambda$-calculus~\cite{EhrhardRegnierTCS03}, an extension to
$\lambda$-calculus with a {\em differential operator} making the
resource-aware behaviour explicit. This extension produces a calculus
where superposition of terms may happen. Then \lalg\ can be seen as
{\em a differential $\lambda$-calculus without the differential
  operator}. In recent years, there has been growing research interest in
these two calculi and their variants, as they could provide an
explicit link between linear logic and linear
algebra~\cite{ArrighiDiazcaroQPL09,ArrighiDiazcaroValironDCM11,DiazcaroPerdrixTassonValironHOR10,DiazcaroPetit10,EhrhardMSCS03,EhrhardMSCS05,EhrhardLICS10,EhrhardRegnierTCS03,PaganiRonchiFOSSACS10,PaganiTranquilliAPLAS09,TassonTLCA09,VauxRTA07}.

The two languages, \llin\ and \lalg, are rather similar: they both
merge the untyped $\lambda$-calculus --higher-order computation in its
simplest and most general form-- with linear algebraic constructions
--sums and scalars subject to the axioms of vector spaces. In both
languages, functions which are linear combinations of terms are
interpreted pointwise: $(\alpha.\ve f + \beta.\ve g)~x = \alpha.(\ve
f~x)+\beta.(\ve g~x)$, where ``$.$'' is the external product. However,
they differ in their treatment of arguments. In \llin, the reduction
strategy is call-by-value (or strictly speaking, call-by-variables or
abstractions) and, in order to deal with the algebraic structure, any
function is considered to be a linear map: $\ve f~(\alpha.x +
\beta.y)$ reduces to $\alpha.(\ve f~x) + \beta.(\ve f~y)$, reflecting
the fact that any quantum evolution is a linear map. On the other
hand, \lalg\ has a call-by-name strategy: $(\lambda x.\,\ve
t)~\ve r$ reduces to $\ve t[\ve r/x]$, with no restrictions on $\ve r$
. As a consequence, the reductions are different as illustrated
by the following example.  In \llin, $(\lambda x.\,x~x)~(\alpha.y +
\beta.z)$ reduces to $\alpha.(y~y) + \beta.(z~z)$ while in \lalg,
$(\lambda x.\,x~x)~(\alpha.y + \beta.z)$ reduces to $(\alpha.y +
\beta.z)~(\alpha.y + \beta.z) = \alpha^2 .(y~y) +
\alpha\times\beta.(y~z) + \beta\times\alpha.(z~y) + \beta^2.(z~z)$.
Nevertheless, they can simulate each other by means of an extension of
the well-known CPS transform that maps call-by-value to call-by-name
and vice versa~\cite{DiazcaroPerdrixTassonValironHOR10}.

Another more fundamental difference between them is the way the
algebraic part of the calculus is treated. In \llin, the algebraic
structure is captured by a rewrite system, whereas in \lalg\ terms are
identified up to algebraic equivalence. Thus, while $\ve t+\ve t$
reduces to $2.\ve t$ in \llin, they are regarded as the same term in
\lalg.  Using a rewrite system allows \llin\ to expose the algebraic
structure in its canonical form, but it is not without some confluence
issues. Consider the term $Y_{\ve b}=(\lambda x.\,\ve b+x~x)~(\lambda
x.\,\ve b+x~x)$. Then $Y_{\ve b}$ reduces to $\ve b+Y_{\ve b}$, so the
term $Y_{\ve b}+Y_{\ve b}$ in \llin\ reduces to $2.Y_{\ve b}$ but also
to $\ve b+Y_{\ve b}+Y_{\ve b}$ and thus to $\ve b+2.Y_{\ve b}$. Note
that $2.Y_{\ve b}$ can only produce an even number of $\ve b$'s
whereas $\ve b+2.Y_{\ve b}$ will only produce an odd number of $\ve
b$'s, breaking confluence. In \lalg{}, on the other hand, $\ve
b+2.Y_{\ve b}=\ve b+Y_{\ve b}+Y_{\ve b}$, solving the problem. The
canonical solution in $\llin$ is to disallow diverging terms. In
\cite{DiazcaroPerdrixTassonValironHOR10} it is assumed that confluence
can be proved in some unspecified way; then, sets of confluent terms
are defined and used in the hypotheses of several theorems that
require confluence. In the original \llin\
paper~\cite{ArrighiDowekRTA08}, certain restrictions are introduced to
the rewrite system, such as having $\alpha.\ve t+\beta.\ve t$ reduce
to $(\alpha+\beta).\ve t$ only when $\ve t$ is in closed normal form.
The rewrite system has been proved locally
confluent~\cite{ArrighiDiazcaroValironDCM11}, so by ensuring strong
normalisation we obtain confluence~\cite{Terese03}. This approach has
been followed in other
works~\cite{ArrighiDiazcaroQPL09,ArrighiDiazcaroValironDCM11,DiazcaroPetit10}
which discuss similar type systems with strong normalisation. While
these type systems give us some information about the terms, they also
impose some undesirable restrictions:
\begin{itemize}
\item In \cite{ArrighiDiazcaroQPL09} two type systems are presented: a
  straightforward extension of System F, which only allows typing $\ve
  t+\ve r$ when both $\ve t$ and $\ve r$ have the same type, and a
  type system with scalars in the types, which keep track of the
  scalars in the terms, but is unable to lift the previous
  restriction.
\item In \cite{DiazcaroPetit10} a type system solving the previous
  issue that can be interpreted in System F is introduced. However,
  it only considers the additive fragment of \llin: scalars are removed
  from the calculus, considerably simplifying the rewrite system.
 \item In \cite{ArrighiDiazcaroValironDCM11} 
   a combination of the two previous approaches is set up: a type system 
   where the types can be weighted and added together is devised. While this is 
   a novel approach, the introduction of type-level scalars makes it 
   difficult to relate it to System F or any other well-known theory.
\end{itemize}
In this paper, we propose an algebraic $\lambda$-calculus featuring
term-rewriting semantics and a type system strong enough to prove
confluence, while remaining expressive and retaining the
interpretation in System F from previous works. In addition, the type
system provides us with lower bounds for the scalars involved in the
terms.

\paragraph{Outline.}
In section~\ref{sec:calculus} the typed version of \llin, called \CA,
is presented.  Section~\ref{sec:properties} is devoted to proving that
the system possesses some basic properties, namely subject reduction
and strong normalisation, which entails the confluence of the
calculus.  Section~\ref{sec:abstractinterpretation} shows an abstract
interpretation of \CA\ into \additive, the additive fragment of \llin.
Finally, section~\ref{sec:conclusion} concludes.

\section{The Calculus}\label{sec:calculus}
We introduce the calculus \CA, which extends explicit System
F~\cite{ReynoldsPS74} with linear combinations of
$\lambda$-terms. Table~\ref{tab:language} shows the abstract syntax of
types and terms of the calculus, where the terms are based on those of
\llin~\cite{ArrighiDowekRTA08}.
Our choice of explicit System F instead of a Curry style
presentation~\cite{ArrighiDiazcaroQPL09,DiazcaroPetit10} stems from
the fact that, as shown in~\cite{ArrighiDiazcaroValironDCM11},
the ``factorisation'' reduction rules (\cf Group F in
Table~\ref{tab:reduction}) in a Curry style setting introduce some
imprecisions.

\begin{table}[t]
\caption{Types and Terms of \CA}
\label{tab:language}
\vspace{0.5em}
\hspace{0.5em}
\begin{minipage}[l]{0.4\linewidth}
\textit{Types:}
\begin{grammar}
  [(colon){::$=$}]
  [(semicolon){$|$}]
  [(period){.}]
  [(nonterminal){}{}]
$\mathsf{<T>}$ : $\mathsf{<U>}$ ; $\mathsf{<T>} + \mathsf{<T>}$ ; $\bar{0}$ \\
$\mathsf{<U>}$ : $X$ ; $\mathsf{<U>} \to \mathsf{<T>}$ ;
$\forall X . \, \mathsf{<U>}$  
\end{grammar}
\end{minipage}
\hspace{0.1em}
\begin{minipage}[r]{0.5\linewidth}
\textit{Terms:}
\begin{grammar}
  [(colon){::$=$}]
  [(semicolon){$|$}]
  [(period){.}]
  [(nonterminal){\bf}{}]
<t> : <b> ; <t> <t> ; <t>@$\mathsf{U}$ ; <0> ; $\alpha.$<t> ; <t> $+$ <t> \\
<b> : $x$ ; $\lambda x  \colon\mathsf{U} .$ <t> ;
$\Lambda X.$ <t> 
\end{grammar}
\end{minipage}

\vspace{0.5em}
\end{table}

We use the convention that abstraction binds as far to the right as
possible and that application binds more strongly than sums and
scalar multiplication.  However, we will freely add parentheses whenever confusion
might arise. Metavariables $\ve{t}, \ve{r}$, $\ve{s}$, $\ve u$, and
$\ve v$ will range over terms.

Terms known as \emph{basis} terms (nonterminal $\ve{b}$ in
Table~\ref{tab:language}) are the only ones that can substitute a variable
in a $\beta$-reduction step.
This ``call-by-$\ve b$''\footnote{
The set of terms in $\ve b$ is not the set of values of $\CA$ (see
Section~\ref{sec:SN}), so technically it
is not ``call-by-value''.
}
strategy plays an important role when interacting with the linearity
from linear-algebra, \eg the term $(\lambda x:U.\,x~x)~(y+z)$ may
reduce to $(y+z)~(y+z)$ and this to $y~y+y~z+z~y+z~z$ in a call-by-name
setting, however if we decide that abstractions should behave as
linear maps, then this call-by-$\ve b$ strategy can be used and the
previous term will reduce to $(\lambda x:U.\,x~x)~y+(\lambda x:U.\,x~x)~z$
and then to $y~y+z~z$. 

For the same reason, we also make a distinction between \emph{unit}
types (nonterminal $\mathsf{U}$ in Table~\ref{tab:language}) and general
types. Unit types cannot include sums of types except in the codomain
of a function type, and they contain all types of System F. General
types are either sums of unit types or the special type
$\bar{0}$. Basis terms can only be assigned unit types. Scalars
(denoted by greek letters) are nonnegative real numbers.  There are no
scalars at the type level, but we introduce the following notation:
for an integer $n\geq 0$, we will write $n.T$ for the type $T + T +
\cdots + T$ ($n$ times), considering $0.T=\bar{0}$. We may also use
the summation symbol $\sum_{i=1}^n T_i$, with $\sum_{i=1}^0
T_i=\bar{0}$. Metavariables $T, R$, and $S$ will range over general
types and $U,V$, and $W$ over unit types.

Table~\ref{tab:reduction} defines the 
term-rewriting system (TRS)
for \CA, which
 consists of directed versions of the
vector-space axioms and $\beta$-reduction for both kinds of
abstractions. All reductions are performed modulo associativity and
commutativity of the $+$ operator.  It is essentially the TRS of
\llin\ \cite{ArrighiDowekRTA08}, 
with an extra
type-application rule. As usual,
$\to^*$ denotes the reflexive transitive closure of the reduction
relation $\to$.

\begin{table}[t]
\caption{One-step Reduction Relation $\to$}
\label{tab:reduction}
$$
\begin{array}{l@{\hspace{1.5em}}l@{\hspace{1.5em}}l}
\textit{Group E:} & \textit{Group F:} & \textit{Group A:} \\

\ve{u} + \ve{0} \to \ve{u} & \alpha.\ve{u} + \beta.\ve{u} \to (\alpha + \beta).\ve{u}  & (\ve{u}+\ve{v})\, \ve{w} \to \ve{u}\, \ve{w} + \ve{v}\, \ve{w} \\

0.\ve{u} \to \ve{0} & \alpha.\ve{u} + \ve{u} \to (\alpha + 1).\ve{u} & \ve{w}\, (\ve{u}+\ve{v}) \to \ve{w}\, \ve{u} + \ve{w}\, \ve{v} \\
1.\ve{u} \to \ve{u} & \ve{u}+\ve{u} \to 2.\ve{u} & (\alpha.\ve{u})\, \ve{v} \to \alpha.(\ve{u}\, \ve{v}) \\

\alpha.\ve{0} \to \ve{0} & \textit{$\beta$-reduction:} & \ve{v}\, (\alpha.\ve{u}) \to \alpha.(\ve{v}\, \ve{u}) \\

\alpha.(\beta.\ve{u}) \to (\alpha\times\beta).\ve{u} & (\lambda x : U.\, \ve{t})\, \ve{b} \to \ve{t}[\ve{b}/x] & \ve{0}\, \ve{u} \to \ve{0} \\

\alpha.(\ve{u}+\ve{v})\to \alpha.\ve{u} + \alpha.\ve{v} & (\Lambda X . \ve{t}) @ U \to \ve{t}[U/X] & \ve{u}\, \ve{0} \to \ve{0}
\end{array}
$$

\begin{mathpar}
\inferrule{\ve{t}\to\ve{t'}}{\ve{t}+\ve{r}\to\ve{t'}+\ve{r}}
\and
\inferrule{\ve{t}\to\ve{t'}}{\alpha . \ve{t}\to\alpha.\ve{t'}}
\and
\inferrule{\ve{t}\to\ve{t'}}{\ve{t}\,\ve{r}\to\ve{t'}\,\ve{r}}
\and
\inferrule{\ve{r}\to\ve{r'}}{\ve{t}\,\ve{r}\to\ve{t}\,\ve{r'}}\\

\inferrule{\ve{t}\to\ve{t'}}{\ve{t}@T\to\ve{t'}@T} 
\and
\inferrule{\ve{t}\to\ve{t'}}{\lambda x : U.\,\ve{t}\to\lambda x : U.\,\ve{t'}}
\and
\inferrule{\ve{t}\to\ve{t'}}{\Lambda X.\ve{t}\to\Lambda X.\ve{t'}}

\end{mathpar}
\end{table}

Substitution for term and type variables (written $\ve{t}[\ve{b}/x]$
and $\ve{t}[U/X]$, respectively) are defined in the usual way to avoid
variable capture.  Substitution behaves like a linear operator when
acting on linear combinations, \eg $(\alpha.\ve{t} +
\beta.\ve{r})[\ve{b}/x] = \alpha.\ve{t}[\ve{b}/x] +
\beta.\ve{r}[\ve{b}/x]$.

Table~\ref{tab:tyrules} defines the notion of type equivalence and shows
the typing rules for the system. The typing judgement
$\Gamma\vdash\ve{t} : T$ means that the term \ve{t} can be assigned
type $T$ in the context $\Gamma$, with the usual definition of typing
context from System F. As a consequence of the design decision of only
allowing basis terms to substitute variables in a $\beta$-reduction,
typing contexts bind term variables to unit types.

\begin{table}[t]
\caption{\CA\ Type Equivalence and Typing Rules}
\label{tab:tyrules}
\textit{Type Equivalence:}$\quad$Equivalence is the least congruence $\equiv$ s.t.
\[
\begin{array}{lll}
T + \bar{0} \equiv T, \qquad & T + R \equiv R + T, \qquad & T + (R + S) \equiv (T + R) + S
\end{array}
\]

\textit{Typing rules:}
\begin{mathpar}
\inferrule{ }{\Gamma, x : U \vdash x : U} \; \textsc{ax}
\and
\inferrule{ }{\Gamma\vdash \mathbf{0} : \bar{0}} \;
\textsc{ax}_{\bar{0}} \\

\inferrule{\Gamma \vdash \mathbf{t} : \sum_{i=1}^\alpha (U \to T_i) \\
  \Gamma \vdash \mathbf{r} : \beta.U}{\Gamma \vdash \mathbf{t}\,
  \mathbf{r} : \sum_{i=1}^\alpha (\beta.T_i)} \; \to_\textsc{E}
\and
\inferrule{\Gamma, x : U \vdash \mathbf{t} : T}{\Gamma \vdash \lambda
  x : U .\, \mathbf{t} : U \to T} \; \to_\textsc{I} \\

\inferrule{\Gamma\vdash\mathbf{t} : \forall X.U}{\Gamma\vdash
  \mathbf{t} @ V: U[V/X]} \; \forall_E
\and
\inferrule{\Gamma\vdash \mathbf{t} : U \and X \notin
\mbox{FV}(\Gamma)}{\Gamma\vdash \Lambda X . \mathbf{t} :
  \forall X.U} \; \forall_I \\

\inferrule{\Gamma\vdash \mathbf{t} : T \\ \Gamma\vdash
  \mathbf{r} : R}{\Gamma\vdash \mathbf{t} + \mathbf{r} : T+R} \; \textsc{+I}
\and
\inferrule{\Gamma\vdash \mathbf{t} : T}{\Gamma\vdash
  \alpha.\mathbf{t}
  : \lfloor\alpha\rfloor.T} \; \textsc{sI} \\

\inferrule{\Gamma\vdash \mathbf{t} : T \and T \equiv R}{\Gamma\vdash
  \mathbf{t} : R} \; \textsc{Eq}
\end{mathpar}
\end{table}

Using standard arrow elimination instead of rule
$\to_E$ would restrict the calculus, since it would force $\ve t$ to
be sum of arrows of the same type $U\to T$. The same would happen
with the argument type $U$: for the term $(\ve t_1+\ve t_2)~(\ve
r_1+\ve r_2)$ to be well-typed, $\ve t_1$ and $\ve t_2$ would need to
have the same type, and also $\ve r_1$ and $\ve r_2$.

In the rule $\to_E$ presented in Table~\ref{tab:tyrules} we relax this
restriction and we allow to have different $T$'s. Continuing with the
example, this allows $\ve t_1$ and $\ve t_2$ to have different types,
provided that they are arrows with the same source type $U$.

\begin{example}\label{ex:flecha-elim}
  Let $\Gamma\vdash\ve{b}_1\type U$, $\Gamma\vdash\ve{b}_2\type U$, $\Gamma\vdash\lambda x.\,\ve{t}\type U\to T$ and $\Gamma\vdash\lambda y.\,\ve{r}\type U\to R$. Then
  $$\prooftree
	\Gamma\vdash(\lambda x.\,\ve{t})+(\lambda y.\,\ve{r})\type(U\to T)+(U\to R)
	\quad
	\Gamma\vdash\ve{b}_1+\ve{b}_2\type U+U
  \justifies\Gamma\vdash((\lambda x.\,\ve{t})+(\lambda y.\,\ve{r}))~(\ve{b}_1+\ve{b}_2)\type T+T+R+R
  \using\to_E
  \endprooftree$$
  Notice that
  $((\lambda x.\,\ve{t})+(\lambda y.\,\ve{r}))~(\ve{b}_1 +\ve{b}_2 )\to^*
  \underbrace{(\lambda x.\,\ve{t})~\ve{b}_1}_{T} +
  \underbrace{(\lambda x.\,\ve{t})~\ve{b}_2}_{T} +
  \underbrace{(\lambda y.\,\ve{r})~\ve{b}_1}_{R} +
  \underbrace{(\lambda y.\,\ve{r})~\ve{b}_2}_{R} $
\end{example}

On the other hand, allowing different $U$'s is sightly more complex:
on account of the distributive rules (Group A) it is required that all
the arrows in the first addend start with a type which has to be the
type of all the addends in the second term. For example, if the given
term is $(\ve t+\ve r)~(\ve b_1+\ve b_2)$, the terms $\ve t$ and $\ve
r$ have to be able to receive both $\ve b_1$ and $\ve b_2$ as
arguments. This could be done by taking advantage of polymorphism, but
the arrow-elimination rule would become much more complex since it
would have to do both arrow-elimination and forall-elimination at the
same time. Although this approach has been shown to be
viable~\cite{ArrighiDiazcaroValironDCM11}, we delay the modification of
the rule to future work, and keep the simpler but more restricted
version, which is enough for the aims of the present paper.

The main novelty of the calculus is its treatment of scalars (rule
\textsc{sI}). In order to avoid having scalars at the type level, when
typing $\alpha.\ve{t}$ we take the floor of the term-level scalar
$\alpha$ and assign the type $\lfloor\alpha\rfloor.T$ to the term,
which is a sum of $T$s.  The intuitive
interpretation is that a type $n.T$ provides a lower bound for the
``amount'' of $\ve{t}:T$ in the term.

The rest of the rules are straightforward. The $\forall_E$
and $\forall_I$ rules enforce the restriction that only unit types can
participate in type abstraction and type application.

\section{Properties}\label{sec:properties}
\subsection{Subject Reduction with Imprecise Types}\label{sec:SR}

A basic soundness property in a typed calculus is the guarantee that
types will be preserved by reduction. However, in $\CA$ types are
imprecise about the ``amount'' of each type in a term.  For example,
let $\Gamma\vdash\ve{t} : T$ and consider the term $\ve s=
(0.9).\ve{t} + (1.1).\ve{t}$. We see that $\Gamma\vdash \ve{s} : T$
and $\ve{s} \to^* 2.\ve{t}$, but $\Gamma\vdash 2.\ve{t} : T+T$. In
this example a term with type $T$ reduces to a term with type $T+T$,
proving that strict subject reduction does not hold for $\CA$.
Nevertheless, we prove a similar property: as reduction progresses,
types are either preserved or \emph{strengthened}, \ie they become
more precise according to the relation $\pce$ (\cf
Table~\ref{tab:pcedef}). This entails that the derived type for a
term is a lower-bound (with respect to $\pce$) for the actual type of
the reduced term.

\begin{table}[t]
 \centering
  \begin{tabular}{c@{\hspace{0.5cm}}c@{\hspace{0.5cm}}c}
  &&\\
  \prooftree{\alpha \leq\beta}
  \justifies{\alpha.T\pce\beta.T}
  \using \textsc{\small Sub-Wk}
  \endprooftree
  &
  \prooftree{T\equiv R}
  \justifies{T\pce R}
  \using\textsc{\small Sub-Eq}
  \endprooftree
  &
  \prooftree{T\pce S \quad S\pce R}
  \justifies{T\pce R}
  \using\textsc{\small Sub-Tr}
  \endprooftree\\
  &&\\
  \prooftree{T_1\pce T_2 \quad S_1\pce S_2}
  \justifies{T_1+S_1\pce T_2+S_2}
  \using\textsc{\small Sub-Ctxt}_1
  \endprooftree
  &
  \prooftree{U_2\pce U_1 \quad T_1\pce T_2}
  \justifies{U_1\to T_1 \pce U_2 \to T_2}
  \using\textsc{\small Sub-Ctxt}_2
  \endprooftree
  &
  \prooftree{T\pce R}
  \justifies{\forall X.T\pce \forall X.R}
  \using\textsc{\small Sub-Ctxt}_3
  \endprooftree\\
  &&\\
  \end{tabular}
  \caption{Inductive definition of the relation $\pce$, where $\leq$ is the ordering of real numbers}
  \label{tab:pcedef}
\end{table}

\begin{theorem}[Subject Reduction up to $\pce$]\label{thm:SR} 
  For any terms $\ve{t}$ and $\ve{t'}$, context $\Gamma$ and type $T$,
  if $\ve{t} \to \ve{t'}$ and $\Gamma\vdash \ve{t} : T$ then there
  exists some type $R$ such that $\Gamma\vdash \ve{t'} : R$ and $T\pce
  R$, where the relation $\pce$ is inductively defined in
  Table~\ref{tab:pcedef}.
\end{theorem}

Intuitively, $T\pce R$ ($R$ is at least as precise as $T$) means that
there are more summands of the same type in $R$ than in $T$, \eg
$A\pce A+A$ for a fixed type $A$. Note that $\pce$ is not the trivial
order relation: although $T\pce T + R$ for any $R$ (because $T\equiv
T+0.R\pce T + 1.R\equiv T + R$), type $T$ cannot disappear from the
sum; if $T\pce S$, then $T$ will always appear at least once in $S$
(and possibly more than once).

The proof of this theorem requires several preliminary lemmas. We give
the most important of them and some details about the proof of the
theorem.

\renewcommand{\theenumi}{\arabic{enumi}}
\renewcommand{\labelenumi}{\theenumi.}

\begin{lemma}[Generation lemmas] \label{GenLemmas} Let $T$ be a type
  and $\Gamma$ a typing context.

  \begin{enumerate}
  \item
    \label{GenApp} For arbitrary terms $\ve u$ and $\ve v$, if $\Gamma\vdash
    {\ve{u}\, \ve{v}} : {T}$, then there exist natural numbers
    $\alpha, \beta$, and types $U\in \mathsf{U}, T_1, \ldots,
    T_\alpha \in \mathsf{T}$, such that
  $\Gamma\vdash  {\ve{u}} : {\sum_{i=1}^\alpha (U \rightarrow T_i)}$ and 
  $\Gamma\vdash {\ve{v}} : {\beta.U}$ with
  $\sum_{i=1}^\alpha (\beta.T_i) \equiv T$.
\item \label{GenAbs} For any term $\ve t$ and unit type $U$, if
  $\Gamma\vdash {\lambda x : U .\, \ve{t}} : {T}$, then there exists a
  type $R$ such that $\Gamma, x : U\vdash {\ve{t}} : {R}$ and $U\to R
  \equiv T$.

\item \label{GenSum} For any terms $\ve u$ and $\ve v$, if
  $\Gamma\vdash {\ve{u}+\ve{v}} : {T}$, then there exist types $R$ and
  $S$ such that
  $\Gamma\vdash {\ve{u}} : {R}$ and
  $\Gamma\vdash {\ve{v}} : {S}$, with
  $R+S \equiv T$.

\item \label{GenScal} For any term $\ve u$ and nonnegative real number
  $\alpha$, if $\Gamma\vdash {\alpha. \ve{u}} : {T}$, then there
  exists a type $R$ such that $\Gamma\vdash {\ve{u}} : {R}$ and
  $\lfloor\alpha\rfloor . R \equiv T$.

\item \label{GenTyAbs} For any term $\ve t$, if $\Gamma\vdash {\Lambda
    X. \ve{t}} : {T}$, then there exists a type $R$ such that
  $\Gamma\vdash {\ve{t}} : {R}$ and $\forall X.R\equiv T$ with
  $X\notin\text{FV}(\Gamma)$.

\item \label{GenTyApp} For any term $\ve t$ and unit type $U$, if
  $\Gamma\vdash {\ve{t}@U} : {T}$, then there exists a type $V$ such
  that $\Gamma\vdash {\ve{t}} : {\forall X.V}$ and $V[U/X]\equiv
  T$.\qed
\end{enumerate}
\end{lemma}

\noindent The following lemma is standard in proofs of subject reduction for
System F-like systems~\cite{Krivine90,Barendregt92}. It ensures that
well-typedness is preserved under substitution on type and term
variables.

\begin{lemma}[Substitution lemma] \label{SubstLemma} For any term $\ve
  t$, basis term $\ve b$, context $\Gamma$, unit type $U$ and 
  type $T$,
  \begin{enumerate}
  \item If $\Gamma\vdash {\ve{t}} : {T}$, then $\Gamma[U/X] \vdash
    \ve{t}[U/X] : T[U/X]$.
  \item If $\Gamma,x:U \vdash {\ve{t}} : {T}$ and $\Gamma\vdash
    {\ve{b}} : {U}$, then $\Gamma\vdash {\ve{t}[\ve{b}/x]} : {T}$.\qed
  \end{enumerate}
\end{lemma}

\noindent Now we can give some details about the proof of Theorem~\ref{thm:SR}.
\begin{proof}[Proof of Theorem~\ref{thm:SR} (Subject Reduction up to $\pce$)] By
  structural induction on the derivation of $\ve t\to\ve t'$. We check that every
  reduction rule preserves the type up to the relation $\pce$. In each case, we
  first apply one or more generation lemmas to the left-hand side of
  the rule. Then we construct a type for the right-hand side which is either
  more precise (in the sense of relation $\pce$) or equivalent to that of the left-hand side.

For illustration purposes, we show the proof of the case corresponding
to the rewrite rule $\alpha.\ve{t} + \beta.\ve{t}\to(\alpha +
\beta).\ve{t}$.

We must prove that for any term $\ve t$, nonnegative real numbers
$\alpha$ and $\beta$, context $\Gamma$ and type $T$, if $\Gamma\vdash
{\alpha . \ve{t} + \beta . \ve{t}} : {T}$ then $\Gamma\vdash {(\alpha
  + \beta) . \ve{t}} : {R}$ with $T\pce R$.

  By lemma \ref{GenLemmas}.\ref{GenSum}, there exist $T_1, T_2$ such that
    $\Gamma\vdash{\alpha . \ve{t}} : {T_1}$ and 
    $\Gamma\vdash{\beta . \ve{t}} : {T_2}$, with 
    $T_1 + T_2 \equiv T$.
    Also by lemma \ref{GenLemmas}.\ref{GenScal}, there exist $R_1,
    R_2$ such that $\Gamma\vdash {\ve{t}} : {R_1}$ with 
    $\lfloor\alpha\rfloor . R_1 \equiv T_1$, and $\Gamma\vdash{\ve{t}}
    :{R_2}$ with 
    $\lfloor\beta\rfloor . R_2 \equiv T_2$.  Then from
    $\Gamma\vdash{\ve{t}} : {R_1}$ we can derive the sequent
    $\Gamma\vdash{(\alpha+\beta).\ve{t}} : {\lfloor\alpha +
      \beta\rfloor. R_1}$ using rule $\textsc{sI}$.

 We will now prove that $T\pce \lfloor\alpha + \beta\rfloor . R_1$. Since
 $R_1$ and $R_2$ are both types for $\ve{t}$, we have $R_1\equiv R_2$
 so
$\lfloor\alpha + \beta\rfloor . R_1
\sce (\lfloor\alpha\rfloor + \lfloor\beta\rfloor) . R_1 
 \equiv 
 \lfloor\alpha\rfloor . R_1 + \lfloor\beta\rfloor . R_1
 \equiv
 \lfloor\alpha\rfloor . R_1 + \lfloor\beta\rfloor . R_2  
 \equiv
 T_1 + T_2 
 \equiv
 T$.
Therefore, we conclude $T\pce \lfloor\alpha + \beta\rfloor . R_1$.\qed
\end{proof}

\subsection{Strong Normalisation}\label{sec:SN}

In this section, we prove the strong normalisation property for
\CA. That is, we show that all possible reductions for well-typed
terms are finite.  We use the standard notion of \emph{reducibility
  candidates}~\cite[Chapter 14]{Girard89}, extended to account for
linear combinations of terms. Confluence follows as a
corollary. Notice that we cannot reuse the proofs of previous typed
versions of \llin\ (\eg \cite{ArrighiDiazcaroQPL09,DiazcaroPetit10})
since in \cite{ArrighiDiazcaroQPL09} only terms of the same type can
be added together, and in \cite{DiazcaroPetit10} the calculus under
consideration is a fragment of \CA. Therefore, none of them have the
same set of terms as \CA.

A closed term in \CA\ is a \emph{value} if it is an abstraction, a sum
of values or a scalar multiplied by a value, \ie values are closed
terms that conform to the following grammar:

\begin{grammar}
  [(colon){::$=$}]
  [(semicolon){$|$}]
  [(period){.}]
  [(nonterminal){\bf}{}]
<v> : $\lambda x \colon U . <t>$ ; $\Lambda X . <t>$ ; <v> $+$ <v> ; $\alpha.$<v>
\end{grammar}

If a closed term is not a value, it is said to be \emph{neutral}. A term
\ve{t} is \emph{normal} if it has no reducts, \ie there is no term
\ve{s} such that $\ve{t}\to\ve{s}$. A \emph{normal form} for a term
\ve{t} is a normal term \ve{t'} such that $\ve{t} \to^* \ve{t'}$.
We define $\text{Red}(\ve{t})$ as the set of reducts of $\ve{t}$
reachable in one step. 

A term \ve{t} is \emph{strongly normalising} if there are no infinite
reduction sequences starting from \ve{t}. We write $\mathsf{SN}_0$ for
the set of strongly normalising closed terms of \CA.

\begin{definition}[Reducibility candidates]
  A set of terms $\mathsf{A}$ is a \emph{reducibility candidate} if
  it satisfies the following conditions:

  \begin{description}
  \item[(CR$_1$)] \emph{Strong normalisation}: $\mathsf{A} \subseteq
    \mathsf{SN}_0$
  \item[(CR$_2$)] \emph{Stability under reduction}: If
    $\ve{t}\in\mathsf{A}$ and $\ve{t}\to^*\ve{t'}$, then
    $\ve{t'}\in\mathsf{A}$.
  \item[(CR$_3$)] \emph{Stability under neutral expansion}: If \ve{t}
    is neutral and $\text{Red}(\ve{t}) \subseteq \mathsf{A}$, then
    $\ve{t}\in\mathsf{A}$.
  \end{description}

\noindent  In the sequel, $\mathsf{A}$, $\mathsf{B}$ stand for reducibility
  candidates, and $\mathsf{RC}$ stands for the set of all reducibility
  candidates.
\end{definition}

The idea of the strong normalisation proof is to interpret types
by reducibility candidates and then show that whenever a term has a
type, it is in a reducibility candidate.

\begin{remark}
  Note that $\mathsf{SN}_0$ is a reducibility candidate. In addition, the term \ve{0} is a neutral and normal term, so it is in every reducibility candidate. This ensures that every reducibility candidate is inhabited, and since every typable term can be closed by typing rule $\to_I$, it is enough to consider only closed terms.
\end{remark}

The following lemma ensures that the strong normalisation property is
preserved by linear combination.

\begin{lemma}\label{lemmSNcloLC}
  If $\ve{t}$ and $\ve{r}$ are strongly normalising, then
  $\alpha.\ve{t}+\beta.\ve{r}$ is strongly normalising.
\end{lemma}

\begin{proof} Induction on a positive algebraic measure defined on
  terms of $\llin$~\cite[Proposition 10]{ArrighiDowekRTA08}, showing that every
  algebraic reduction makes this number strictly decrease. \qed
\end{proof}

The following operators 
ensure that all types of \CA\ are interpreted by a reducibility
candidate.

\begin{definition}[Operators in $\mathsf{RC}$]\label{def:operatorsRC}
  Let $\mathsf{A}$, $\mathsf{B}$ be reducibility candidates. We
  define operators $\to$, $\oplus$, $\Lambda$ over $\mathsf{RC}$ and $\overline{\emptyset}$ such that

  \begin{itemize}
  \item $\mathsf{A}\to\mathsf{B}$ is the closure of
    $\{\ve{t}\mid\forall\ve{b}\in\mathsf{A},\, \ve{b} \text{ a basis
      term } \Rightarrow (\ve{t})\,\ve{b}\in\mathsf{B}\}$ under
    (CR$_3$),
  \item $\mathsf{A}\oplus\mathsf{B}$ is the closure of $\{
    \alpha.\ve{t}+\beta.\ve{r} \mid \ve{t}\in\mathsf{A},
    \ve{r}\in\mathsf{B} \}$ under (CR$_2$) and (CR$_3$),
  \item $\Lambda \mathsf{A}$ is the set $\{ \ve{t} \mid
    \forall V,  \ve{t}@V\in\mathsf{A} \}$
  \item $\overline{\emptyset}$ is the closure of $\emptyset$ under (CR$_3$).
  \end{itemize}
\end{definition}

\begin{remark}
Notice that $0$ is neutral and it is in normal form. Therefore the closure of $\emptyset$ under (CR$_3$) is not empty, it includes, at least, the term $0$.
\end{remark}

\begin{lemma}\label{lemm:operatorsRC}
  Let $\mathsf{A}$ and $\mathsf{B}$ be reducibility candidates. Then
  $\mathsf{A\to B}$, $\mathsf{A\oplus B}$, $\Lambda \mathsf{A}$, $A \cap B$ and $\overline{\emptyset}$
  are all reducibility candidates.
\end{lemma}

\begin{proof} We show the proof for $\overline{\emptyset}$ and
  $\mathsf{A\oplus B}$. The rest of the cases are similar.

\begin{itemize}
 \item The three conditions hold trivially for $\overline{\emptyset}$.
 \item Let $\ve{t}\in\mathsf{A\oplus B}$. We must check that the
     three conditions hold.
     \begin{description}
     \item[(CR$_1$)] Induction on the construction of
       $\mathsf{A\oplus B}$. If
       $\ve{t}\in\{\alpha.\ve{t}+\beta.\ve{r}\mid\ve{t}\in
       \mathsf{A}, \ve{r}\in\mathsf{B} \}$, the result is trivial by
       condition (CR$_1$) on $\mathsf{A}$ and $\mathsf{B}$ and
       lemma~\ref{lemmSNcloLC}. If $\ve{t}\to^*\ve{t'}$ with
       $\ve{t}\in\mathsf{A\oplus B}$, then \ve{t} is strongly
       normalising by induction hypothesis; therefore, so is
       \ve{t'}. If \ve{t} is neutral and
       $\text{Red}(\ve{t})\subseteq\mathsf{A\oplus B}$, then \ve{t}
       is strongly normalising since by induction hypothesis all
       elements of $\text{Red}(\ve{t})$ are strongly normalising.
      
      \item[(CR$_2$) and (CR$_3$)] Trivial by construction of
        $\mathsf{A\oplus B}$. \qed
      \end{description}
\end{itemize}
  
\end{proof}

\newcommand{\interp}[1]{\ensuremath{\llbracket #1 \rrbracket}}

\noindent We can now introduce the interpretation function for the types of
\CA. The definition relies on the operators for reducibility
candidates defined above.

A \emph{valuation} $\rho$ is a partial function from type variables to
reducibility candidates, written as a sequence of comma-separated
mappings of the form $X \mapsto \mathsf{A}$, with $\emptyset$
denoting the empty valuation.

\begin{definition}[Reducibility model] Let $T$ be a type and $\rho$ a
  valuation. We define the \emph{interpretation} $\interp{T}_\rho$ as follows:
\[
\begin{array}{rcl}
  \interp{X}_\rho & = & \rho(X) \\
  \interp{\bar{0}}_\rho & = & \overline{\emptyset} \\
  \interp{U\to T}_\rho & = & \interp{U}_\rho \to \interp{T}_\rho \\
  \interp{T + R}_\rho & = & \interp{T}_\rho \oplus \interp{R}_\rho \\
  \interp{\forall X.U}_\rho & = & \bigcap_{\mathsf{S}\in\mathsf{RC}} \Lambda \interp{U}_{\rho,X\mapsto S}
\end{array}
\]
\end{definition}

Note that lemma \ref{lemm:operatorsRC} ensures that every type is
interpreted by a reducibility candidate.

A \emph{substitution} $\sigma$ is a partial function from term
variables to basis terms, written as a sequence of semicolon-separated
mappings of the form $x\mapsto \ve{b}$, with $\emptyset$ denoting the
empty substitution. The action of substitutions on terms is given by
$$\ve{t}_\emptyset = \ve{t}, \qquad\qquad \ve{t}_{x\mapsto \ve{b};\sigma} = \ve{t}[\ve{b}/x]_\sigma$$

A \emph{type substitution} $\delta$ is a partial function from type
variables to unit types, written as a sequence of semicolon-separated
mappings of the form $X\mapsto U$, with $\emptyset$ denoting the
empty substitution. The action of type substitutions on types is given by
$$T_\emptyset = T, \qquad\qquad T_{X\mapsto U;\delta} = T[U/X]_\delta$$
They are extended to act on terms in the natural way.

Let $\Gamma$ be a typing context, then we say that a substitution pair
$\langle\sigma,\delta\rangle$ \emph{satisfies} $\Gamma$ for a valuation $\rho$ (written
$\langle\sigma,\delta\rangle\in\interp{\Gamma}_\rho$) if $(x : U) \in\Gamma$ implies
$x_\sigma\in\interp{U_\delta}_\rho$.

A typing judgement $\Gamma\vdash {\ve{t}} : {T}$ is said to be
\emph{valid} (written $\Gamma \vDash \ve{t} : T$) if for every
valuation $\rho$, for every type substitution $\delta$ and every
substitution $\sigma$ such that $\langle\sigma,\delta\rangle\in\interp{\Gamma}_\rho$, we have
$(\ve{t}_{\delta})_\sigma \in\interp{T}_\rho$.
The following lemma proves that every derivable typing judgement is valid.

\begin{lemma}[Adequacy Lemma]\label{adequacyLemma}
  Let $\Gamma\vdash {\ve{t}} : {T}$, then $\Gamma\vDash\ve{t}:T$.
\end{lemma}

\begin{proof}
  We proceed by induction on the derivation of $\Gamma\vdash\ve{t} :
  T$. The base cases (rules \textsc{Ax} and \textsc{Ax}$_{\bar{0}}$)
  are trivial. We show the cases for rules $\to_I$ and \textsc{sI} for
  illustration purposes.  

\begin{itemize}
\item Case $\to_I$: $\inferrule{\Gamma,x:U\vdash \ve{t} :
    T}{\Gamma\vdash\lambda x:U.\ve{t} : U\to T}$

  By induction hypothesis, we have $\Gamma,x:U\vDash \ve{t} : T$. We
  will prove that for all $\rho$ and
  $\langle\sigma,\delta\rangle\in\interp{\Gamma}_\rho$, $((\lambda
  x:U.\ve{t})_\delta)_\sigma\in\interp{U\to T}_\rho$. Suppose that $\sigma =
  (x\mapsto \ve{v};\sigma'')\in\interp{\Gamma,x:U}_\rho$. Let
  $\ve{b}\in\interp{U}_\rho$ (note that there is at least one basis
  term, $\ve{v}$, in $\interp{U}_\rho$), and let $\sigma'=(x\mapsto
  \ve b;\sigma'')$. So $\sigma'\in\interp{\Gamma, x : U}_\rho$, hence
  $(\ve{t}_\delta)_{\sigma'}\in\interp{T}_\rho$. This means both \ve{b} and
  $(\ve{t}_\delta)_{\sigma'}$ are strongly normalising, so we shall first prove
  that all reducts of $((\lambda x:U.\ve{t})_\delta)_\sigma\,\ve{b}$ are in
  $\interp{T}_\rho$.

    \begin{itemize}
    \item $((\lambda x:U.\ve{t})_\delta)_\sigma\,\ve{b}\to(\lambda
      x:U.\ve{t'})\,\ve{b}$ or $((\lambda x:U.\ve{t})_\delta)_\sigma\,\ve{b}\to
      ((\lambda x:U.\ve{t})_\delta)_\sigma\,\ve{b'}$, with
      $(\ve{t}_\delta)_\sigma\to\ve{t'}$ or $\ve{b}\to\ve{b'}$.  The result
      follows by induction on the reductions of $\ve{b}$ and
      $(\ve{t}_\delta)_{\sigma'}$, respectively: by induction hypothesis we
      have $\ve{t'},\ve{b'}\in\interp{T}_\rho$, so both $(\lambda
      x:U.\ve{t'})\,\ve{b}$, $(((\lambda
      x:U.\ve{t})_\delta)_\sigma)\,\ve{b'}\in\interp{T}_\rho$.

    \item $((\lambda
      x:U.\ve{t})_\delta)_\sigma\,\ve{b}\to(\ve{t}_\delta)_\sigma[\ve{b}/x]
      = (\ve{t}_\delta)_{\sigma'}\in\interp{T}_\rho$.
    \end{itemize}

    Therefore, $((\lambda x:U.\ve{t})_\delta)_\sigma \,\ve{b}$ is a neutral
    term with all of its reducts in $\interp{T}_\rho$, so  $((\lambda
    x:U.\ve{t})_\delta)_\sigma \,\ve{b}\in\interp{T}_\rho$. Hence, by
    definition of $\to$, we conclude $((\lambda
    x:U.\ve{t})_\delta)_\sigma\in\interp{U\to T}_\rho$.

\item Case \textsc{sI}:
  $\inferrule{\Gamma\vdash\ve{t} : T}{\Gamma\vdash \alpha.\ve{t} : \lfloor\alpha\rfloor.T}$

  By induction hypothesis, we have $\Gamma\vDash\ve{t} : T$.  Let
  $\rho$ be a valuation and $\langle\sigma,\delta\rangle$ a
  substitution pair satisfying $\Gamma$ in $\rho$. So
  $(\ve{t}_\delta)_\sigma\in\interp{T}_\rho$, hence
  $\alpha.(\ve{t}_\delta)_\sigma\in\bigoplus_{i=1}^{\lfloor\alpha\rfloor}\interp{T}_\rho
  = \interp{\lfloor\alpha\rfloor . T}_\rho$ by construction. \qed
\end{itemize}
\end{proof}

\noindent Since this proves that every well-typed term is in a reducibility
candidate, we can easily show that such terms are strongly normalising.

\begin{theorem}[Strong Normalisation for \CA]\label{thm:SN}
  All typable terms of \CA\ are strongly normalising.
\end{theorem}

\begin{proof}
  Let \ve{t} be a term of \CA\ of type $T$. If $\ve t$ is an open term, the open variables are in the context, so we can always close it and the term will be closed and typable. Then we can consider $\ve t$ to be closed. Then, by the Adequacy
  Lemma (lemma~\ref{adequacyLemma}), we know that $(\ve{t}_\emptyset)_\emptyset
  \in \interp{T}_\emptyset$. Furthermore, by lemma~\ref{lemm:operatorsRC},
  we know $\interp{T}_\emptyset$ is a reducibility candidate, and
  therefore $\interp{T}_\emptyset \subseteq \mathsf{SN}_0$. Hence,
  \ve{t} is strongly normalising.\qed 
\end{proof}

\subsubsection{Confluence}\label{sec:conf}

Now confluence follows as a corollary of the strong normalisation theorem.

\begin{corollary}[Confluence]\label{cor:confluenceADD}
 The typed language \CA\ is confluent: for any term $\ve t$, if $\ve t\to^*\ve r$ and $\ve t\to^*\ve u$, then there exists a term $\ve t'$ such that $\ve r\to^*\ve t'$ and $\ve u\to^*\ve t'$.
\end{corollary}
\begin{proof}
  The proof of the local confluence of the system, \ie the property
  saying that $\ve t\to\ve r$ and $\ve t\to\ve u$ imply that there
  exists a term $\ve t'$ such that $\ve r\to^*\ve t'$ and $\ve
  u\to^*\ve t'$, is an extension of the one presented for the untyped
  calculus in \cite{ArrighiDiazcaroValironDCM11}, where the set of
  algebraic rules (\ie all rules but the beta reductions) have
  been proved to be locally confluent using the proof assistant
  Coq. Then, a straightforward induction entails the (local) commutation
  between the algebraic rules and the ${\beta}$-reductions. Finally,
  the confluence of the $\beta$-reductions is a trivial extension of
  the proof for $\lambda$-calculus.
  Local confluence plus strong normalisation (\cf Theorem~\ref{thm:SN}) implies confluence~\cite{Terese03}. \qed
\end{proof}

\section{Abstract Interpretation}\label{sec:abstractinterpretation}
The type system of $\CA$ approximates the more precise types that are
obtained under reduction. The approximation suggests that a
$\lambda$-calculus without scalars can be seen as an abstract
interpretation of \CA: its terms can approximate the terms of \CA.
Scalars can be approximated to their floor, and hence be represented
by sums, just as the types in \CA\ do. This intuition is formalised
in this section, using \additive, the calculus
presented in \cite{DiazcaroPetit10}. This calculus is a typed version
of the additive fragment of \llin~\cite{ArrighiDowekRTA08}, which in
turn is the untyped version of \CA.

The \additive\ calculus is shown in Table~\ref{tab:Additive}. It
features strong normalisation, subject reduction and confluence. For
details on those proofs, please refer to \cite{DiazcaroPetit10}. The
types and equivalences coincide with those from \CA. We write the
types explicitly in the terms to match the presentation of \CA,
although the original presentation is in Curry style.  We use
\vdashadd\ to distinguish the judgements in \CA\ ($\vdash$) from the
judgements in \additive. Also, we write the reductions in \additive\
as \toA, $\ve t\toAnormal$ for the normal form of the term $\ve t$ in
\additive\ and $\ve{t}\tonormal$ for the normal form of \ve{t} in \CA.

\begin{table}
  \caption{The \additive\ calculus. Type syntax, equivalences and type rules
    coincide with those from \CA, except for rule $sI$ which does not exist which is not necessary in this calculus.}
\label{tab:AddLanguage}
\vspace{0.5em}
\hspace{0.5em}
\begin{minipage}[l]{0.98\linewidth}
	$$\begin{array}[t]{l@{\hspace{1.5cm}}r@{~::=~}l}
    \text{\itshape Terms:} & \ve{t},\ve r & \ve{b}~|~\ve{t}~\ve{r}~|~\ve t@U~|~\ve{0}~|~\ve{t}+\ve{r}\\
    \text{\itshape Basis terms:} & \ve{b} & x~|~\lambda x:U.\,\ve{t}~|~\Lambda X.\ve t
    \end{array}$$
    \begin{center}
	\noindent\begin{tabular}{p{4,5cm}p{2,5cm}p{4cm}}  
    \noindent \emph{Group A:}
    
    \noindent $(\ve{u}+\ve{t})~\ve{r}\toA \ve{u}~\ve{r}+\ve{t}~\ve{r}$
    
    \noindent $(\ve{r})~(\ve{u}+\ve{t})\toA\ve{r}~\ve{u}+\ve{r}~\ve{t}$

    \noindent $\ve{0}~\ve{t}\toA \ve{0}$
    
    \noindent $\ve{t}~\ve{0}\toA \ve{0}$

    &
    
    \noindent \emph{Group E:}

    \noindent $\ve{t} + \ve{0}\toA \ve{t}$
    &
    \noindent \emph{$\beta$-reduction:}
    
    \noindent $(\lambda x:U.\,\ve{t})~\ve{b}\toA\ve{t}[\ve b/x]$

	\noindent $(\Lambda X.\ve t)@U\toA\ve t[U/X]$
  \end{tabular}
  \end{center}
  \label{tab:Additive}
\end{minipage}
\vspace{0.5em}
\end{table}

Let $T_c$ be the set of terms in the calculus $c$. Consider the following {\em abstraction} function $\sigma : \TFA \to
 \Tadd$ from terms in \CA\ to terms in \additive:
$$\begin{array}{r@{~=~}l@{\hspace{1cm}}r@{~=~}l}
	\sigma(x:U) & x:U	&		\sigma(\ve t@U) & \sigma(\ve t)@U\\
	\sigma(\lambda x{\type}U.\,\ve t) & \lambda x{\type}U.\,\sigma(\ve t)	&	\sigma(\ve 0) & \ve 0\\
	\sigma(\Lambda X.\ve t) & \Lambda X.\sigma(\ve t)	&	\sigma(\alpha.\ve t) & \sum_{i=1}^{\lfloor\alpha\rfloor} \sigma(\ve t)\\
	\sigma(\ve t~\ve t') & \sigma(\ve t)~\sigma(\ve t')	&	\sigma(\ve t+\ve t') & \sigma(\ve t)+\sigma(\ve t')
  \end{array}$$
where for any term $\ve t$, $\sum_{i=1}^0\ve t=\ve 0$.

We can also define a {\em concretisation} function $\gamma:
\Tadd\to\TFA$, which is the obvious embedding of terms: $\gamma(\ve
t)=\ve t$.

Let $\sqsAstrict\;\subseteq\;\Tadd\times\Tadd$ be the least relation satisfying:
$$\begin{array}{r@{~\Rightarrow~}l@{\hspace{1cm}}r@{~\Rightarrow~}l}
 \multicolumn{4}{c}{\alpha\leq\beta~\Rightarrow~\sum_{i=1}^{\alpha}\ve t\sqsAstrict \sum_{i=1}^{\beta}\ve t} \\
 \ve t\sqsAstrict\ve t' &\lambda x{\type}U.\,\ve t\sqsAstrict\lambda x{\type}U.\,\ve t' &
 \ve t\sqsAstrict\ve t'~\wedge~\ve r\sqsAstrict\ve r' & (\ve t)~\ve r\sqsAstrict (\ve t')~\ve r' \\
 \ve t\sqsAstrict\ve t' &\Lambda X.\ve t\sqsAstrict\Lambda X.\ve t' &
 \ve t\sqsAstrict\ve t'~\wedge~\ve r\sqsAstrict\ve r' & \ve t+\ve r\sqsAstrict\ve t'+\ve r' \\
 \ve t\sqsAstrict\ve t' & \ve t@U\sqsAstrict\ve t'@U &
\ve t\sqsAstrict \ve r\phantom{'}~\wedge~\ve r\sqsAstrict \ve s\phantom{'} & \ve{t}\sqsAstrict \ve{s}
\end{array}$$
and let $\sqsA$ be the relation defined by 
$\ve t_1\sqsA\ve t_2\Leftrightarrow \fnormalA{\ve t_1}\sqsAstrict \fnormalA{\ve t_2}$.

The relation \sqsAstrict\ is a partial order. Also, \sqsA\ is a partial order if we quotient terms by the relation $\sim$, defined by $\ve t\sim\ve r$ if and only if $\fnormal{\ve t}=\fnormal{\ve r}$ . We formalise this in the following lemma.

\begin{lemma}\label{lem:poset}~
\begin{enumerate}
 \item $\sqsAstrict$ is a partial order relation
 \item $\sqsA$ is a partial order relation in $\Tadd/_\sim$. \qed
\end{enumerate}
\end{lemma}
\noindent The following theorem states that the terms in \CA\ can be seen as a refinement of those in \additive, \ie we can consider \additive\ as an abstract interpretation of \CA. It follows by a nontrivial structural induction on $\ve t\in\TFA$.

\begin{theorem}[Abstract interpretation]\label{thm:AbsIntFA-Ad}
The function $\tonormal$ is a valid concretisation of the function $\toAnormal$: $\forall\ve t\in\TFA$, $\fnormalA{\sigma(\ve t)}\sqsA\sigma(\fnormal{\ve t})$. \qed
\end{theorem}

The following lemma states that the abstraction preserves the typings.

\begin{lemma}\label{lem:typingcorrectnessFA-Ad}
For arbitrary context $\Gamma$, term $\ve t$ and type $T$, if
$\Gamma\vdash\ve t\type T$ then $\Gamma\vdashadd \sigma(\ve t)\type T$. \qed
\end{lemma}

Taking \additive\ as an abstract interpretation of \CA\ entails the
extension of the interpretation of \additive\ into System F with pairs, 
\Fp\ (\cf~\cite{DiazcaroPetit10}) as an abstract interpretation of \CA, as depicted in Figure~\ref{fig:AI}.
The complete language \Fp\ is defined in Table~\ref{tab:Fp}. We denote by $\fnormalFp{t}$ the normal form of a term $t$ in \Fp. The relation \sqsFp\ is a straightforward translation of the relation $\sqsA$ into a relation in \Fp. The function \AddtoFpD{\cdot}\ is the translation from typed terms in \additive\ into terms in \Fp; this translation depends on the typing derivation $\mathsf{D}$ of the term in \additive\ (\cf \cite{DiazcaroPetit10} for more details).
We formalise this in Theorem~\ref{thm:AbsIntFA-Fp} and also give the formal definition of the relation \sqsFp\ in definition \ref{def:relFp}.

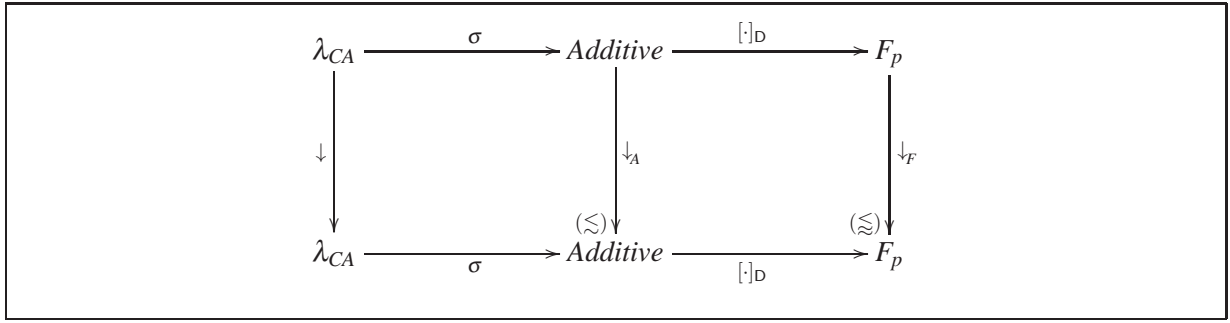
\begin{figure}[t]
  $$\xymatrix@C=15ex@R=12ex{
    \CA
    \ar@{->}[d]_{\tonormal}
    \ar@{->}[r]^{\sigma}
    & 
    \additive
    \ar@{->}[d]^{\:\toAnormal}_(0.85){(\sqsA)}
	\ar@{->}[r]^{\AddtoFpD{\cdot}}
	& 
	\Fp
	\ar@{->}[d]^{\toFpnormal}_(0.85){(\sqsFp)}
    \\
    \CA
    \ar@{->}[r]_{\sigma}
    &
    \additive
	\ar@{->}[r]_{\AddtoFpD{\cdot}}
	& 
	\Fp
  }$$
\caption{Abstract interpretation of \CA\ into System F with pairs}
 \label{fig:AI}

\end{figure}

\begin{table}[ht]
\caption{System F with pairs}
 \label{tab:Fp}
\vspace{0.5em}
\hspace{0.5em}
\begin{minipage}[l]{0.98\linewidth}
	$$\begin{array}[t]{l@{\hspace{1.5cm}}r@{~::=~}l}
      \text{\itshape Terms:} & t,u & x~|~\lambda x.\,t~|~tu~|~\star~|~\pair t u~|~\pi_1(t)~|~\pi_2(t)\\
      \text{\itshape Types:} & A,B & X~|~A \toF B~|~\forall X.A~|~\uno~|~A\times B\\
    \end{array}$$
	$$\begin{array}{c@{\qquad;\qquad}c}
	(\lambda x.\,t)u\toF t[u/x] &\pi_i\pair{t_1}{t_2}\toF t_i 
	  \end{array}$$
  \begin{mathpar}
    \inferrule{ }{\Delta,x:A\thesi x:A}{\scriptstyle Ax}
    \and
    \inferrule{ }{\Delta\thesi \star:\uno}{\scriptstyle \uno}
    \and
    \inferrule{\Delta,x:A\thesi t:B}{\Delta\thesi \lambda x.\,t:A\toF B}
    {\scriptstyle \toF I}
    \and
    \inferrule{\Delta\thesi t:A\toF B \and \Delta\thesi u:A}{\Delta\thesi tu:B}
    {\scriptstyle \toF E}
    \and
    \inferrule{\Delta\thesi t:A\and \Delta\thesi u:B}{\Delta\thesi\pair{t}{u}:A\times B}
    {\scriptstyle \times  I}
    \and
    \inferrule{\Delta\thesi t:A\times B}{\Delta\thesi\pi_1(t):A}{\scriptstyle\times  E_{\ell}}
    \and
    \inferrule{\Delta\thesi t:A\times B}{\Delta\thesi\pi_2(t):B}{\scriptstyle\times  E_r}
    \\
    \inferrule{\Delta\thesi t:A\and X\notin F\!V(\Delta)}{\Delta\thesi t:\forall X.A}
    {\scriptstyle \forall I}
    \and
    \inferrule{\Delta\thesi t:\forall X.A}{\Delta\thesi t:A[B/X]}{\scriptstyle \forall E}
  \end{mathpar}
\end{minipage}
\end{table}

\begin{definition}\label{def:relFp}
  Let $\sqsFpstrict\subseteq \TFp\times\TFp$ be the least relation
  between terms of \Fp\ satisfying:
$$\begin{array}{r@{~\Rightarrow~}lcr@{~\Rightarrow~}l}
 \multicolumn{2}{c}{	\star\sqsFpstrict  t} & t\sqsFpstrict t
	& \multicolumn{2}{c}{ t\sqsFpstrict( t, t)} \\
  t\sqsFpstrict t'~\wedge~ r\sqsFpstrict r' & ( t, r)\sqsFpstrict( t', r') &
	&  t\sqsFpstrict t' &\lambda x.\, t\sqsFpstrict\lambda x.\, t' \\
  t\sqsFpstrict t'~\wedge~ r\sqsFpstrict r' &  t~ r\sqsFpstrict  t'~ r' &
	&  t\sqsFpstrict t' & \pi_1( t)\sqsFpstrict\pi_1( t') \\
  t\sqsFpstrict r\phantom{'}~\wedge~ r\sqsFpstrict  s\phantom{'} & {t}\sqsFpstrict {s} &
	&  t\sqsFpstrict t' & \pi_2( t)\sqsFpstrict\pi_2( t')
\end{array}
$$
and let $\sqsFp$ be the relation defined by $ t_1\sqsFp t_2\Leftrightarrow \fnormalA{ t_1}\sqsFpstrict \fnormalA{ t_2}$.
\end{definition}

The relation $\sqsFpstrict$ is a partial order. Moreover \sqsFp\ is a 
partial order if we quotient terms in $\Fp$ by the equivalence relation
$\approx$, defined as: $t\approx r$ if and only if $\fnormalFp{t}=\fnormalFp{r}$.
\begin{lemma}\label{lem:Fppartialorder}~
\begin{enumerate}
 \item $\sqsFpstrict$ is a partial order relation.
 \item \sqsFp\ is a partial order relation over  $\TFp/_\approx$. \qed
\end{enumerate}
\end{lemma}

\noindent In \cite[Thm. 3.8]{DiazcaroPetit10} it is shown that the translation
\AddtoFpD{\cdot}\ is well behaved. So it will trivially keep the
order.

\begin{lemma}\label{lem:translationToFpKeepsOrder}
  Let $\mathsf{D}$ be a derivation tree ending in $\Gamma\vdashadd\ve
  t: T$ and $\mathsf{D'}$ be a derivation tree corresponding to
  $\Gamma\vdashadd\ve r: R$, where $\ve t\sqsA\ve r$. Then
  $\AddtoFpD{\ve t}\sqsFp\AddtoFp{\ve r}{D'}$. \qed
\end{lemma}

\begin{theorem}\label{thm:AbsIntFA-Fp} The function $\tonormal$ is a
  valid concretisation of $\toFpnormal$: $\forall\ve t\in\TFA$ if
  $\mathsf{D}$ is a derivation of $\Gamma\vdash\sigma(\ve t):T$ and
  $\mathsf{D'}$ is the derivation of $\Gamma\vdash\sigma(\fnormal{\ve
    t}):T'$, then $\fnormalFp{\AddtoFpD{\sigma(\ve t)}}
  \sqsFp\AddtoFp{\sigma(\fnormal{\ve t})}{D'}$.
\end{theorem}
\begin{proof}
  Theorem \ref{thm:AbsIntFA-Ad} states that the left
  square in Figure~\ref{fig:AI} commutes,
  lemma~\ref{lem:typingcorrectnessFA-Ad} states that the typing is
  preserved by this translation, and finally
  lemma~\ref{lem:translationToFpKeepsOrder} states that
  the square on the right commutes.\qed
\end{proof}

\section{Summary of Contributions}\label{sec:conclusion}
We have presented a confluent, typed, strongly normalising, algebraic
$\lambda$-calculus, based on \llin, which has an algebraic rewrite
system without restrictions. Typing guarantees confluence, thereby
allowing us to simplify the rewrite rules for the system with respect
to $\llin$. Moreover, $\CA$ differs from $\lalg$ in that it presents
vectors in a canonical form by using a rewrite system instead of an
equational theory.

In this work, scalars are approximated by natural numbers.  This
approximation yields a subject reduction property which is exact about
the types involved in a term, but only approximate in their ``amount''
or ``weight''. In addition, the approximation is a lower bound: if a
term has a type that is a sum of some amount of different types, then
after reducing it these amounts can be incremented but never
decremented.

One of the original motivations for this work was to ensure confluence
in the presence of algebraic rewrite rules, while remaining
``classic'', in the sense that the type system does not introduce
uninterpretable elements, \ie elements that cannot have an exact interpretation in a classical system, such as scalars. To prove that we have
achieved this goal, we have shown that terms in \additive, the
additive fragment of \llin, can be seen as an abstract interpretation
of terms in \CA, and then System F can also be used as an abstract
interpretation of terms in \CA\ by the translation from \additive\
into \Fp.

In our calculus, we have chosen to take the floor of the scalars to
approximate types. However, this decision is arbitrary, and we could
have chosen to approximate types using the ceiling instead. Therefore,
an obvious extension of this system is to take both floor and ceiling
of scalars to produce type intervals, thus obtaining more accurate
approximations.

An interesting suggestion made for one of the reviewers is to use truth values instead of natural numbers, which although will loose precision in the interpretation (indeed, it would be as interpreting all non-zero values by $1$) could make the interpretation into a classical system much more direct.

Since this paper is meant as a ``proof of concept'' we have not worked
around a known restriction in \additive, which allows sums as
arguments only when all their constituent terms have the same
type, \eg $\ve t~(\ve r+\ve s)$ cannot have a type unless $\ve r$ and $\ve s$ have the same type. However, it has been proved that this can be solved by using a
more sophisticated arrow elimination typing
rule~\cite{ArrighiDiazcaroValironDCM11}.

Since the type system derives from System F, there are some total
functions which cannot be represented in \CA, even though they are
expressible in \llin. This is not a problem in practice because these
functions are quite hard to find, so it is a small price to pay for
having a simpler, confluent rewrite system.

It is still an open question how to obtain a similar result for a
calculus where scalars are members of an arbitrary ring.

\section*{Acknowledgements}
We would like to thank Pablo Arrighi, Philippe Jorrand, Simon Perdrix,
Barbara Petit, and Beno\^it Valiron for enlightening discussions.
This work was supported by grants from DIGITEO and R\'egion \^Ile-de-France, and also by the CNRS--INS2I PEPS project QuAND.

\bibliographystyle{eptcs}
\bibliography{biblio-full}

\end{document}